\documentclass[11pt, a4paper]{amsart}

\usepackage{a4wide}

\usepackage[utf8]{inputenc}
\usepackage[T1]{fontenc}
\usepackage{subfigure, url}

\usepackage[round]{natbib}

\usepackage{amssymb,amscd,amsfonts,mathrsfs,amsmath}
\usepackage[all]{xy}
\usepackage{color}

\usepackage{tikz}
\usetikzlibrary{arrows,shapes,positioning}
\usetikzlibrary{decorations.markings}
\tikzstyle arrowstyle=[scale=1]
\tikzstyle directed=[postaction={decorate,decoration={markings,
    mark=at position .5 with {\arrow[arrowstyle]{stealth}}}}]
\newcommand{\ra}{\rightarrow}
\newcommand{\e}{\varepsilon}

\newcommand{\Ra}{\Rightarrow}
\newcommand{\cyc}{{cyc}}

\newcommand\A{A}
\newcommand\B{B}
\newcommand\C{C}

\newcommand{\X}{X}
\newcommand{\Y}{Y}
\newcommand{\Z}{Z}

\newcommand{\s}{S}
\newcommand{\cA}{\mathcal{A}}
\newcommand{\cT}{\mathcal{T}}
\newcommand{\cF}{\mathcal{F}}

\newcommand{\Aho}{MR0258547}

\newcommand{\HU}{MR645539}
\newcommand{\Brandstadt}{MR630064}
\newcommand{\Oshiba}{MR0478788}
\newcommand{\Maslov}{MR0334597}

\newcommand{\ERnumber}{MR0363010}

\newtheorem{thm}{Theorem}[section]
\newtheorem{prop}[thm]{Proposition}
\newtheorem{lem}[thm]{Lemma}

\newtheorem{defn}[thm]{Definition}

\usepackage{prettyref}
\newcommand{\prref}[1]{\prettyref{#1}}
\newif{\ifshort}\shorttrue
\shortfalse
\ifshort
\newrefformat{alg}{Alg.~\ref{#1}}
 \newrefformat{thm}{Thm.~\ref{#1}}
 \newrefformat{lem}{Lem.~\ref{#1}}
 \newrefformat{def}{Def.~\ref{#1}}
 \newrefformat{cor}{Cor.~\ref{#1}} 
 \newrefformat{prop}{Prop.~\ref{#1}}
 \newrefformat{sec}{Sect.~\ref{#1}}
 \newrefformat{kap}{Chap.~\ref{#1}}
 \newrefformat{ex}{Ex.~\ref{#1}}
 \newrefformat{app}{App.~\ref{#1}}
 \newrefformat{alg}{Alg.~\ref{#1}}
 \newrefformat{eq}{(\ref{#1})}%
 \newrefformat{tab}{Tab.~\ref{#1}}
 \newrefformat{fig}{Fig.~\ref{#1}}
 \newrefformat{rem}{Rem.~\ref{#1}}
 \newrefformat{ap}{{\sc Appendix}}

\else
\newrefformat{alg}{Algorithm~\ref{#1}}
\newrefformat{thm}{Theorem~\ref{#1}}
\newrefformat{lem}{Lemma~\ref{#1}}
\newrefformat{def}{Definition~\ref{#1}}
\newrefformat{cor}{Corollary~\ref{#1}}
\newrefformat{prop}{Proposition~\ref{#1}}
\newrefformat{sec}{Section~\ref{#1}}
\newrefformat{subsec}{Subsection~\ref{#1}}
\newrefformat{kap}{Chapter~\ref{#1}}
\newrefformat{ex}{Example~\ref{#1}}
\newrefformat{rem}{Remark~\ref{#1}}
\newrefformat{fig}{Figure~\ref{#1}}
\newrefformat{app}{Appendix~\ref{#1}}
\newrefformat{eq}{Equation~(\ref{#1})}
\newrefformat{tab}{Table~\ref{#1}}
\newrefformat{ap}{{\sc Appendix}}

\begin{document}
\title[Permutations of  context-free,  ET0L and indexed languages]{Permutations of  context-free,  ET0L \\and indexed languages}

\author[T. Brough]{Tara Brough}
\address{Universidade de Lisboa, Portugal}
\email{tarabrough@gmail.com}

\author[L. Ciobanu]{Laura Ciobanu}
\address{University of Neuch\^atel, Switzerland}
\email{laura.ciobanu@unine.ch}

\author[M. Elder]{Murray Elder}
\address{
The University of Newcastle, Australia}
\email{murray.elder@newcastle.edu.au}

\author[G. Zetzsche]{Georg Zetzsche}
\address{  LSV, CNRS \& ENS Cachan, Universit\'{e} Paris-Saclay, France}
\email{zetzsche@cs.uni-kl.de}

\keywords{ET0L, EDT0L, indexed, context-free,    cyclic closure}
\subjclass[2010]{20F65; 68Q45}
\date{May 2016}
\thanks{Research supported by  London Mathematical Society Scheme 4 grant 41348, Swiss National Science 
Foundation Professorship FN PP00P2-144681/1,  Australian Research Council grant  FT110100178, and  the Postdoc-Program of the German Academic Exchange Service (DAAD)}

\begin{abstract}
For a language $L$, we consider its cyclic closure, and more generally the language $C^k(L)$, which consists of all words obtained by 
partitioning words from $L$ into $k$ factors and permuting them.
We prove that the classes of ET0L and EDT0L languages are closed under the 
operators $C^k$.  This both sharpens and generalises Brandst\"adt's result that if $L$ is context-free 
then $C^k(L)$ is context-sensitive 
 and not context-free  in general for $k\geq 3$.
We also show  that the cyclic closure 
of an indexed language is indexed. 
\end{abstract}

\maketitle

\section{Introduction}
\label{sec:in}
In this note we investigate closure properties of context-free, ET0L, EDT0L and indexed languages under the operation of permuting a finite number of factors.
Let $S_k$ denote the set of permutations on $k$ letters.
We sharpen a result of \cite{\Brandstadt}  who proved that if $L$ is context-free (respectively one-counter, linear) then the language \[C^k(L)=\{w_{\sigma(1)}\ldots w_{\sigma(k)}\mid w_1\ldots w_k\in L,\sigma\in S_k\}\]
 is not context-free (respectively one-counter, linear) in general for $k\geq 3$.
In our main result, Theorem \ref{thm:etolmain}, we prove that if $L$ is ET0L (respectively EDT0L), then $C^k(L)$ is also ET0L (respectively EDT0L). Since context-free languages are ET0L, it follows that if $L$ is context-free,  then  $C^k(L)$ is ET0L. 
 \cite{\Brandstadt}  proved that regular, context-sensitive and recursively enumerable languages are closed under $C^k$, so our results extend this list to include ET0L and EDT0L.

The language $C^2(L)$ is simply the 
 {\em cyclic closure} of $L$, given by
\[\cyc(L)=\{w_2w_1\mid w_1w_2\in L\}.\] 
\cite{\Maslov, \Oshiba}
proved that the cyclic closure of a context-free language is context-free. In Theorem \ref{thm:firstmain} we show that the same is true for indexed languages.

The cyclic closure of a language, as well as the generalization $C^k$, are natural operations on languages, which can prove useful in determining whether a language belongs to a certain class. These operations are particularly relevant when studying languages attached to conjugacy in groups and semigroups (see \cite{CHHR}).

\section{Permutations of ET0L and EDT0L languages}

The acronym {ET0L} (respectively {EDT0L}) refers to {\em {\rm E}xtended, {\rm T}able, 
{\rm 0} interaction, and {\rm L}indenmayer} (respectively D\emph{eterministic}). 
There is a vast literature on Lindenmayer systems, see \cite{RozS86}, with various  acronyms such as D0L, DT0L, ET0L, HDT0L and so forth. 
The following inclusions hold: EDT0L $\subset$ ET0L $\subset$ indexed, and context-free $\subset$ ET0L.
Furthermore, the classes of EDT0L and 
context-free languages are incomparable. 

\begin{defn}[ET0L]
An \emph{ET0L-system} is a tuple $H=(V,\cA,\Delta,I)$, where 
\begin{enumerate}
\item $V$ is a finite alphabet,
\item $\cA\subseteq V$ is the subset of \emph{terminal symbols},
\item $\Delta=\{P_1,\ldots,P_n\}$ is a finite set of \emph{tables}, meaning
each $P_i$ is a finite subset 
of $V\times V^*$, and
\item $I\subseteq V^*$ is a finite set of \emph{axioms}.
\end{enumerate}

A word over $V$ is called a \emph{sentential form (of $H$)}.
For $u,v\in V^*$, we write $u\Rightarrow_{H,i} v$ if $u=c_1\cdots c_m$
for some $c_1,\ldots, c_m\in V$ and $v=v_1\cdots v_m$ for some 
$v_1,\ldots,v_m\in V^*$ with $(c_j,v_j) \in P_i$ for every $j\in \{1,\ldots, m\}$.  
We write $u\Rightarrow_H v$
if $u\Rightarrow_{H,i} v$ for some $i\in\{1,\ldots,n\}$.
If there exist sentential forms $u_0,\ldots,u_k$ with $u_i\Rightarrow_{H} u_{i+1}$ for $0\leq i\leq n-1$,
then we write $u_0 \Rightarrow_H^* u_k$.
The language \emph{generated by $H$} is defined as
\begin{align*} L(H) & =\{ v\in \cA^* \mid w\Rightarrow_H^* v~\text{for some $w\in I$} \}. \end{align*}
A language is ET0L if it is equal to $L(H)$ for some ET0L system $H$.
 \end{defn}
 
We may write $c\to v\in P$ to mean $(c,v)\in P$.
We call $(c,v)$ a {\em rule for $c$},
and use the convention that if for some $c\in V$
no rule for $c$ is specified in $P$, then $P$ contains the 
rule $(c,c)$.
 
 \begin{defn}[EDT0L]
An \emph{EDT0L-system}
 is an ET0L system where in each table there is
 exactly one rule for each letter in $V$. 
 A language is EDT0L if it is equal to $L(H)$ for some EDT0L system $H$.
 \end{defn}

In this section we prove the following:
\begin{thm}\label{thm:etolmain}
Let $\cA$ be a finite alphabet.
If $L\subseteq \cA^*$ is ET0L (respectively EDT0L)  then  $C^k(L)$ is ET0L (respectively EDT0L).
\end{thm}
\begin{proof} We start by showing that if $\#_0,\dots, \#_k$  are distinct symbols not in $\cA$ and $L$ is ET0L (respectively EDT0L) then so is  $$L'=\{\#_0w_1\#_1\dots \#_{k-1}w_k\#_k\mid w_1\dots w_k\in L\}.$$ This will be done in Lemma~\ref{lem:inserthash} below.
We then prove in Proposition~\ref{prop:moveright} 
that if $L_1$ is an ET0L (respectively EDT0L) language where each word in $L_1$ has two symbols 
$a,b$ appearing exactly once, 
then $L_2=\{uabwv\mid uavbw \in L_1\}$ is ET0L (respectively EDT0L).
For each permutation $\sigma\in S_k$ we apply this result to $L'$ for $$(a,b)=\left(\#_{\sigma(1)-1},\#_{\sigma(1)}\right), \dots, \left(\#_{\sigma(k)-1},\#_{\sigma(k)}\right)$$  to obtain the ET0L (respectively EDT0L) language 
\begin{align*}L_\sigma&=\{\#_0\#_1\dots \#_kw_{\sigma(1)}\dots w_{\sigma(k)}\mid \#_0w_1\#_1\dots \#_{k-1}w_k\#_k\in L'\}.\end{align*}
We obtain $C^k(L)$ by applying erasing homomorphisms to remove the $\#_i$,  and taking the union over all $\sigma\in S^k$.
Since ET0L (respectively EDT0L) languages are closed under homomorphism and finite union, this shows that $C^k(L)$ is ET0L (respectively EDT0L).

Thus the proof will be complete once we established the above facts.
\end{proof}

\begin{lem}\label{lem:insertonehash}
If $L\subseteq \mathcal A^*$ is EDT0L and $\#$ is a symbol not in $\mathcal A$ then the language
$$L_\#=\{u\#v\mid uv\in L\}$$ is EDT0L.
\end{lem}
\begin{proof}
Let $H=(V,\cA,\Delta,I)$ be an EDT0L system with $L=L(H)$. Without loss of generality we can assume $I \subseteq V$.
Define an EDT0L system $H_\#=(V_\#,\cA\cup\{\#\},\Delta_\#,I_\#)$ as follows: $V_\#$ is the disjoint union $V\cup\{c_\#\mid c\in V\}$, $I_\#=\{s_\#\mid s\in I\}$, and
 $m=\max_{P\in \Delta}\{|w|\mid (c,w)\in P\}$, the length of the longest right-hand side of any table.
Furthermore, we define $\Delta_\#$ to be the disjoint union $\Delta\cup \{P_{i,\#}, P_{\#,i}\mid P\in \Delta, i\in [0,m] \}$, where 
\begin{equation} 
\begin{split}
P_{i,\#}&:=\{c_\# \to ud_\#v \mid c \to udv\in P, |u|=i, d\in V\} \cup P, \\
P_{\#,i}&:=\{c_\# \to u\#v \mid c \to uv\in P, |u|=i\} \cup P.
\end{split}
\end{equation}

We point out that if $c\ra\e\in P$, where $\varepsilon$ denotes the empty word, then $P_{\#,0}=\{c_\#\ra \#\}$, so $\{c_\# \to \# \mid c \to \varepsilon \in P\}$ 
will be included in $\Delta_\#$.

The new system remains finite  since we have added a finite number of new letters and tables, and deterministic since letters $v_\#$ appear exactly once on the left side of each rule in the new tables.

Each word in $L(H_\#)$ is obtained starting with $s_\#\in I_\#$ and applying tables of the form $P_{i,\#}$
some number of times, until at some point, since $\cA\cup\{\#\}$ does not contain any letter with subscript $\#$, 
a table of the form $P_{\#,i}$
must be applied. Before this point there is precisely one letter in the sentential form with subscript $\#$, and
after  there are no letters with subscript $\#$.
Also, if $uv\in L(H)$, then there is some $a\in I$ with $a\Rightarrow_H^* uv$, and by construction $a_\#\Rightarrow_{H_\#}^* u\#v$.
\end{proof}

\begin{lem}\label{lem:inserthash}
If $L\in\cA$ is ET0L (respectively EDT0L)  and $\#_0,\dots, \#_n$ are distinct symbols not in $\mathcal A$, then
$$L'=\{\#_0u_1\#_1\dots u_n\#_n\mid u_1\dots u_n\in L\}$$ is ET0L (respectively EDT0L).
\end{lem}
\begin{proof}
Since ET0L languages are  closed under rational transduction (\cite{RozS86}), the
result is immediate for ET0L. In contrast, the EDT0L languages are not closed under inverse
homomorphism (for example, the language $\{a^{2^n}\mid n\in\mathbb N\}$ is
EDT0L and its inverse homomorphic image $\{w\in\{a,b\}^*\mid \exists n\in\mathbb
N(|w|_a=2^n)\}$ is not (\cite{\ERnumber}, Example 3). Instead, we apply 
Lemma~\ref{lem:insertonehash} $n+1$ times to insert single copies of the $\#_i$,  then intersect
with the regular language $\{\#_0u_1\#_1\dots u_n\#_n\mid  u_i\in \mathcal
A^*\}$ to ensure that the $\#_i$ appear in the correct order.
\end{proof}

 \begin{defn}[$(a,b)$-language]
Let $\cT$ be a finite alphabet 
and $a,b\in \cT$  distinct symbols. We say that $w\in \cT^*$ is an
\emph{$(a,b)$-word} if $w\in X^* a X^* b X^*$, where $X=\cT\setminus \{a,b\}$. A
language $L\subseteq \cT^*$ of $(a,b)$-words is called an \emph{$(a,b)$-language}.

We define a function $\pi$ on $(a,b)$-words as follows. If $w=xaybz\in\cT^*$, then
$\pi(w)=xabzy$. For an $(a,b)$-language $L$, we set $\pi(L)=\{\pi(w) \mid w\in
L\}$. 
\end{defn}

Suppose $L$ is an $(a,b)$-language and $H=(V,\cT, \Delta, I)$ is an ET0L or EDT0L system with $L=L(H)$.
 \begin{defn}[$(a,b)$-morphism]
A morphism $\varphi:V^*\ra \{a,b\}^*$ is called an {\em $(a,b)$-morphism} (for $H$)
if \begin{enumerate}\item[(1)] 
$\varphi(a)=a$, $\varphi(b)=b$, and $\varphi(c)=\varepsilon$ for  $c\in \cT\setminus \{a,b\}$, and
\item[(2)] if $u,v\in V^*$ with $u\Rightarrow_H v$ then $\varphi(u)=\varphi(v)$.
\end{enumerate}
\end{defn}

\begin{lem}\label{lem:abmorphism}
Let $L$ be an ET0L (respectively EDT0L) language that is an $(a,b)$-language.  Then $L$ can be generated by
some ET0L-system (respectively EDT0L-system) that admits an $(a,b)$-morphism.
\end{lem}
\begin{proof}
Suppose $L$ is generated by $H=(V,\cT,\Delta,I)$, where $a,b\in \cT$ and $\Delta=\{P_1,\ldots,P_n\}$. 
Without loss of generality, we may assume that $I\subseteq V$.
We define a new ET0L (respectively EDT0L) system $H'=(V',\cT,\Delta',I')$ 
as follows.
Let $\cF= \{\varepsilon, a, b, ab\}$ be the set of factors of $ab$. 
Let $V'=(V\times\cF)\cup \cT$ be the new alphabet and define the morphism $\varphi\colon V'^*\to \{a,b\}^*$ by
$\varphi((c,f))=f$ for $(c,f)\in V\times \cF$, $\varphi(a)=a$, $\varphi(b)=b$ and
$\varphi(c)=\varepsilon$ for $c\in \cT\setminus \{a,b\}$. 

The role of the $\cF$-component of a symbol $(c,f)$ in $V'$ is to store the
$\varphi$-image of the terminal word to be derived from $c$.  Since $H$ generates
only $(a,b)$-words, we choose as axioms $I'=I\times\{ab\}$. The role of
the tables is to distribute the two letters (in the $\cF$-component) in each
word along a production.

In the ET0L case, the new set of tables is $\Delta'=\{P'_1,\ldots,P'_n,P'_{n+1}\}$, where
\begin{align*}
P'_i &= \{ (c, f) \to (c_1,f_1)\cdots (c_m,f_m) \mid c\to c_1\cdots c_m\in P_i,~~f=f_1\cdots f_m \} \end{align*}
for each $i\in \{1,\ldots,n\}$ and
\begin{align*}
P'_{n+1} = \{ (a, a) \to a,~ (b, b) \to b\} ~~\cup~~ \{ (c,\varepsilon) \to c \mid c\in \cT\setminus\{a,b\} \} ~~\cup~~ \{c\to c \mid c\in \cT \}.
\end{align*}
In the EDT0L case, we introduce a separate table for each choice of a factorisation $f=f_1\cdots f_\ell$ for each $f\in\cF$, where $\ell$ is the maximal length of any right-hand side in $H$.

The idea underlying the definition of the tables $P'_i$ is that we make multiple copies of each rule in $P_i$ based on the choices for how to partition $f$ and distribute the factors among the $c_i$'s.

We claim now that $H'=(V',\cT,\Delta',I')$ admits the morphism $\varphi$. Property (1) follows from the definition of $\varphi$, and property (2) from the definition of the tables above.

Let $\psi\colon V'^*\to V^*$ be the `first coordinate projection' morphism with $\psi((c,f))=c$ for
$(c,f)\in V\times \cF$ and $\psi(c)=c$ for $c\in \cT$.  
 
For the inclusion $L(H')\subseteq L(H)$, note that $u\Rightarrow_{H'} v$
implies $\psi(u)\Rightarrow_H \psi(v)$ or $\psi(u)=\psi(v)$, so in any case
$\psi(u)\Rightarrow_H^*\psi(v)$.  Thus, if $v\in L(H')$ with $w\Rightarrow_{H'}^*
v$ and $w\in I'$, then $\psi(w)\Rightarrow_{H}^*\psi(v)$ and $\psi(w)\in I$,
hence $v=\psi(v)\in L(H)$. This implies $L(H')\subseteq L(H)$.

For the inclusion $L(H)\subseteq L(H')$, a straightforward induction on $n$
yields the following claim: If $u\Rightarrow_H^n v$ with $u\in V^*$ and an
$(a,b)$-word $v\in \cT^*$, then we have $u'\Rightarrow_{H'}^* v$ for some $u'\in
V'^*$ such that $\psi(u')=u$ and $\varphi(u')=ab$. We apply this to a
derivation $s\Rightarrow_H^* v$ with $s\in I$. Then our claim yields an $s'\in
V'^*$ with $s'\Rightarrow_{H'}^* v$, $\psi(s')=s\in I$, and
$\varphi(s')=ab$.  This means $s'\in I'$ and thus $v\in L(H')$.
\end{proof}

\begin{prop}\label{prop:moveright}
Let $L$ be an $(a,b)$-language that is ET0L (respectively EDT0L).  Then $\pi(L)$ is  ET0L (respectively EDT0L).
\end{prop}
\begin{proof}
Let $L=L(H)$, where $H=(V,\cT,\Delta,I)$. By \prref{lem:abmorphism}, we may assume that there is an
$(a,b)$-morphism $\varphi$ for $H$.  We now use $\varphi$ to define a map
similar to $\pi$ on words over $V$.  A word $w\in V^*$ is said to be an
\emph{$(a,b)$-form} (short for $(a,b)$-sentential-form) if $\varphi(w)=ab$. Such a word is either of the
form $xCy$, where $r,s\in V^*$ and $C\in V$, with
$\varphi(x)=\varphi(y)=\varepsilon$ and $\varphi(C)=ab$; or it is of the form
$xAyBz$ with $x,y,z\in V^*$ and $A,B\in V$ with
$\varphi(x)=\varphi(y)=\varphi(z)=\varepsilon$ and $\varphi(A)=a$,
$\varphi(B)=b$. 
In the former case, $w$ is called \emph{fused}, in the latter it is called \emph{split}.

Let $p,q$ be symbols with $p,q\notin V$. We define the function $\tilde{\pi}$
on $(a,b)$-forms as follows.  If $w$ is fused, then
$\tilde{\pi}(w)=wpq$. If $w$ is split with $w=xAyBz$ as above, then
$\tilde{\pi}(w)=xABzpyq$. In other words, 
the factor between $a$ and $b$ in $w$ will be moved between $p$ and $q$.
For a set $L$ of $(a,b)$-forms, we set
$\tilde{\pi}(L)=\{\tilde{\pi}(w) \mid w\in L\}$.  Note that $\tilde{\pi}$ differs from
$\pi$ by introducing the letters $p,q$. This will simplify the ensuing
construction.

The idea is to construct an ET0L (respectively EDT0L) system
$H'=(V',\cT',\Delta',I')$, in which $V'$ is the disjoint union $V\cup \{p,q\}$
and $\cT'=\cT\cup \{p,q\}$, such that for $(a,b)$-forms $u,v\in V^*$, we have
\begin{align} &u\Rightarrow_H v  && \text{if and only if} && \tilde{\pi}(u) \Rightarrow_{H'} \tilde{\pi}(v) \label{eq:eperm} \end{align}
Moreover, for each $(a,b)$-form $u\in V^*$ and $v'\in V'^*$ with
$\tilde{\pi}(u)\Rightarrow_{H'} v'$, there is an $(a,b)$-form $v\in
V^*$ such that 
\begin{equation}\label{eq:fperm} 
\begin{gathered}
 \xymatrix@M=7pt@=20pt{ 
u \ar@2{->}[r]_>{H} \ar@{|->}[d]_{\tilde{\pi}}               & v \ar@{|->}[d]^{\tilde{\pi}}\\
\tilde{\pi}(u) \ar@2{->}[r]_>{H'}   & v'
}
\end{gathered}
\end{equation}
For example, if the derivation $\tilde{\pi}(xAyBz)=xABzpyq \Rightarrow_{H'} x'A'B'z'py'q$ holds (the split-split case for $u$ and $v$), then $xAyBz \Rightarrow_{H}x'A'y'B'z'$, and similar implications hold in the other cases.

We define $I'$ as $I'=\{ \tilde{\pi}(w) \mid w\in I \}$, hence equation (\ref{eq:eperm})
implies $\tilde{\pi}(L(H))\subseteq L(H')$ and equation (\ref{eq:fperm}) implies
$L(H')\subseteq \tilde{\pi}(L(H))$.  Together, we have
$L(H')=\tilde{\pi}(L(H))$, meaning $\tilde{\pi}(L(H))$ is an ET0L (respectively EDT0L) language.
Furthermore, we have $\pi(L(H))=\psi(\tilde{\pi}(L(H)))$, where $\psi$ is the
homomorphism that erases $p,q$. Thus, since the classes of ET0L and EDT0L languages are
closed under homomorphic images, proving equations (\ref{eq:eperm}), (\ref{eq:fperm}) implies that
$\pi(L(H))$ is an ET0L (respectively EDT0L) language and hence \prref{prop:moveright}.

As before, we write $\Delta=\{P_1,\ldots,P_n\}$.
Let $\ell$ be the maximal length of a right-hand side in the productions of $H$, and let $V^{\le \ell}$ denote the set of all words in $V^*$ of length at most $\ell$.
The set $\Delta'$ consists of the following tables:
\begin{align*}
&P'_{i}     && \text{for each $1\le i\le n$}, \\
&P'_{i,w}   && \text{for each $1\le i\le n$ and $w\in V^{\le \ell}$ with $\varphi(w)=\varepsilon$,} \\
&P'_{i,u,v} && \text{for each $1\le i\le n$ and $u,v\in V^{\le \ell}$ with $\varphi(u)=\varphi(v)=\varepsilon$,}
\end{align*}
which we describe next.  The table $P'_i$ allows $H'$ to mimic (in the sense of
\eqref{eq:eperm}) steps in $P_i$ that start in a fused word and result in
a fused word. Each table $P'_i$ comprises the following productions:
\begin{align*}
& A\to z  && \text{for each $A\to z\in P_i$ with $\varphi(A)=\varepsilon$,} \\
& C\to xDy && \text{for each $C\to xDy\in P_i$ with $D\in V$} \\
&          && \text{and $\varphi(C)=\varphi(D)=ab$,} \\
& p \to p, && \\
& q\to q.  &&
\end{align*}
The table $P'_{i,w}$ mimics all steps of $P_i$ where a fused word is
turned into a split one, such that between the introduced $A,B\in V$,
$\varphi(A)=a$, $\varphi(B)=b$, the word $w$ is inserted. It contains the
following productions:
\begin{align*}
& A\to z        && \text{for each $A\to z\in P_i$ with $\varphi(A)=\varepsilon$,} \\
& C\to x A B y  && \text{for each $C\to x A w B y \in P_i$ with $\varphi(C)=ab$,} \\
&               && \text{$\varphi(A)=a$, and $\varphi(B)=b$,} \\
& p \to p w,    && \\
& q\to q.       &&
\end{align*}
Finally, the table $P'_{i,u,v}$ mimics a step of $P_i$ that starts in a split word and produces a split one, such that 
(i) the symbol $A$ with $\varphi(A)=a$ generates $u$ to its right and (ii) the symbol $B$ with $\varphi(B)=b$ generates $v$ to its left.
It consists of the productions
\begin{align*}
& A\to z     && \text{for each $A\to z\in P_i$ with $\varphi(A)=\varepsilon$,} \\
& A \to xA'  && \text{for each $A\to xA'u\in P_{i}$ with $\varphi(A)=\varphi(A')=a$,} \\
& B \to B'y  && \text{for each $B\to vB'y\in P_{i}$ with $\varphi(B)=\varphi(B')=b$,} \\
& p \to p u, && \\
& q\to vq.    && 
\end{align*}

It can be verified straightforwardly that with these tables, equations (\ref{eq:eperm}), (\ref{eq:fperm})
are satisfied. In addition, if the table $P_i$ has exactly one rule for each letter in $V$ then $P'_i,P'_{i,w}$ and $P'_{iu,v}$ has exactly one rule for each letter in $V'$, so if $H$ is EDT0L then so is $H'$. We have thus proven \prref{prop:moveright}.
\end{proof}

\section{Cyclic closure of indexed is indexed}\label{sec:main}

Recall that an indexed language is one that is generated by the following type of grammar:
\begin{defn}
[Indexed grammar; \cite{\Aho}]
An {\em indexed grammar}  is a 5-tuple $(\mathcal N, \mathcal T, \mathcal I, \mathcal P, S)$ such that
\begin{enumerate}
\item $\mathcal N, \mathcal T, \mathcal I$ are three mutually disjoint sets of symbols, called {\em nonterminals, terminals} and {\em indices} (or {\em flags}) respectively.
\item $S\in\mathcal N$ is the {\em start symbol}.
\item $\mathcal P$ is a finite set of {\em productions}, each having the form of one of the following:
\begin{enumerate}
\item  $A \ra B^f$.
\item $A^f \ra v$.
\item  $A \ra u$.
\end{enumerate}
where $A, B \in\mathcal N$, $f\in \mathcal I$ and $u,v\in(\mathcal N\cup\mathcal T)^*$.
\end{enumerate}
\end{defn}
As usual in grammars, indexed grammars successively transform sentential forms,
which are defined as follows. An \emph{atom} is either a terminal letter
$x\in\mathcal{T}$ or a pair $(A,\gamma)$ with $A\in \mathcal N$ and $\gamma\in
\mathcal{I}^*$. Such a pair $(A,\gamma)$ is also denoted $A^\gamma$.  A
\emph{sentential form} of an indexed grammar is a (finite) sequence of atoms.
In particular, every string over $\mathcal T$ is a sentential form. The
language defined by an indexed grammar is the set of all strings of terminals
that can be obtained by successively applying production rules starting from 
the sentential form $S$. Let $\A\in \mathcal N, \gamma\in \mathcal I^*$.
Define a letter homomorphism $\pi_\gamma$ by
\[ 
\pi_\gamma(c)=\begin{cases}c^\gamma & \text{if $c\in \mathcal N$},\\
c & \text{if $c\in \mathcal T$}.\end{cases}
\]
In contrast to ETOL systems, where each step replaces every symbol in the sentential form,
indexed grammars transform only one atom per step.
Production rules transform sentential forms as follows.  
For an atom $A^\gamma$ in the sentential form:
\begin{enumerate}
\item applying $A \ra B^f$ replaces one occurrence of $A^\gamma$ by $B^{f \gamma}$
\item if $\gamma=f\delta$ with $f\in \mathcal I$, applying $A^f\ra v$ replaces one occurrence of $A^\gamma$ (with $\gamma\in\mathcal I^*$) by $\pi_\delta(v)$ 
\item applying $A\ra u$ replaces one occurrence of $A^\gamma$ by $\pi_\gamma(u)$.
\end{enumerate}
We call the operation of successively applying productions starting from the
sentential form $S$ and terminating at a string $u\in\mathcal T^*$  a {\em
derivation} of $u$.  We use the notation $\Rightarrow$ to denote a sequence of
productions within a derivation, and call such a sequence a {\em
subderivation}. Sometimes we abuse notation and write $u\to v$ for sentential
forms $u$ and $v$ to denote that $v$ results from $u$ by applying one rule.

We represent a derivation $S\Rightarrow u\in \mathcal T^*$ pictorially using a {\em parse tree}, 
which is defined in the same way as for context-free grammars (see for example \cite{\HU} page 83) 
with root labeled by $S$, internal nodes labeled by $A^\omega$ for $A\in  \mathcal N$ and $\omega\in \mathcal{I}^*$ 
and leaves labeled by $\mathcal T\cup\{\varepsilon\}$.

A {\em path-skeleton} of a parse tree is the (labeled) $1$-neighbourhood of some path from the root vertex to a leaf.
See \prref{fig:pathskeleton} for an example.

\begin{defn}[Normal form]
An indexed grammar $(\mathcal N, \mathcal T, \mathcal I, \mathcal P, S)$ is in {\em normal form} if
all productions are of one of the following types:
\begin{enumerate}
\item  $A \ra B^f$
\item \label{second}$ A^f \ra B$
\item \label{third} $A \ra BC$
\item  \label{forth} $A \ra a$
\end{enumerate}
where $A, B, C \in\mathcal N$, $f\in \mathcal I$ and $a\in\mathcal T$.
\end{defn}

An indexed grammar can be put into normal form as follows. 
For each production $\A^f \ra v$ with $v\not\in\mathcal N$, introduce a new nonterminal $\B$, add productions  $\A^f \ra \B, \B \ra v$, and remove $\A^f \ra v$.
By the same arguments used for Chomsky normal form, each production  $A \ra u$ without flags can be replaced by a set of 
productions of type \ref{third} and \ref{forth} above.

\cite{\Maslov, \Oshiba}
proved that the cyclic closure of a context-free language is context-free.
A sketch of a proof of this fact is given in the solution to Exercise 6.4 (c) in \cite{\HU},
and we generalise the approach taken there to show that the class of indexed languages is also closed 
under the cyclic closure operation.

\begin{thm}\label{thm:firstmain}
If $L$ is indexed, then $\cyc(L)$ is indexed. 
\end{thm}

\begin{proof}
The idea of the proof is to take the parse-tree of a derivation of $w_1w_2\in L$ in $\Gamma$ and ``turn it upside down",
using the leaf corresponding to the first letter of the word $w_2$ as the new start symbol. 

Let $\Gamma = (\mathcal N, \mathcal T, \mathcal I, \mathcal P, \s)$ be an indexed grammar for $L$ in normal form.
If $w = a_1\ldots a_n\in L$ with $a_i\in \mathcal T$ and we wish to generate the cyclic permutation $a_k\ldots a_n a_1\ldots a_{k-1}$ of $w$, 
take some parse tree for $w$ in $\Gamma$ and draw the unique path $F$ from the start symbol $\s$ to $a_k$.
Consider the path-skeleton for $F$.

\begin{figure}[h!]   
\begin{center}
\begin{tikzpicture}[scale=1]
\draw[ultra thick,decorate]  (2,7) -- (4,5) -- (1,2)--(3,0);

\draw[decorate] (2,7) -- (1,6);
\draw[decorate] (4,5) -- (5,4);
\draw[decorate] (3,4) -- (4,3);
\draw[decorate] (1,2) -- (0,1);

\draw (3,6) node {$\bullet$};
\draw (2,3) node {$\bullet$};
\draw (2,1) node {$\bullet$};

\draw (2.1,7.4) node {$\s$};
\draw (0.7,6.2) node {$\A_1$};
\draw (3.3,6.2) node {$\B_1$};
\draw (4.4,5.2) node {$\B_2^{f}$};
\draw (5.3,4.2) node {$\A_4^{f}$};
\draw (2.7,4.2) node {$\B_3^{f}$};
\draw (4.3,3.2) node {$\A_3^{f}$};
\draw (1.7,3.2) node {$\B_4^{f}$};
\draw (0.7,2.2) node {$\B_5^{gf}$};
\draw (-.3,1.2) node {$\A_2^{gf}$};
\draw (2.5,1.2) node {$\B_6^{gf}$};
\draw (3,-.2) node {$a_k$};

\end{tikzpicture}

\end{center}
\caption{Path-skeleton in an indexed grammar.  
\label{fig:pathskeleton}}
\end{figure}
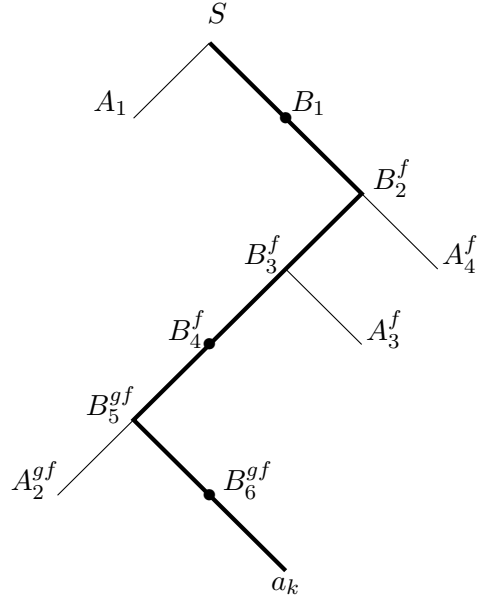

In the example given in \prref{fig:pathskeleton},
the desired word $a_k\ldots a_n a_1\ldots a_{k-1}$ can be derived from the 
string $a_k \A_3^{f} \A_4^{f} \A_1 \A_2^{gf}$, using productions in $\mathcal P$.  

Therefore we wish to enlarge the grammar to generate all strings $$a_k \A_{k+1}^{w_{k+1}}\dots \A_n^{w_n} \A_1^{w_1} \ldots \A_{k-1}^{w_{k-1}},$$
where $\A_1^{w_1},\ldots,\A_{k-1}^{w_{k-1}}$ are the labels of the vertices lying immediately to the left of $F$ (in top to bottom order), 
and $\A_{k+1}^{w_{k+1}},\ldots,\A_n^{w_n}$ are the labels of the vertices lying immediately to the right of $F$ (in bottom to top order).
We do this by introducing new `hatted' nonterminals, with which we label all the vertices along the path $F$,
and new productions which are the reverse of the old productions `with hats on'.  By first nondeterministically guessing
the flag on the nonterminal immediately preceding $a_k$, we are able to essentially generate the path-skeleton in reverse.

The grammar for $\cyc(L)$ is given by $\Gamma' = (\mathcal N', \mathcal T', \mathcal I', \mathcal P\cup \mathcal P', S_0)$, 
where $\mathcal T'=\mathcal T$, $\mathcal I'=\mathcal I\cup \{\$\}$ (where $\$$ is a new symbol not in $\mathcal I$), 
$\s_0\in\mathcal N'\setminus \mathcal N$ is the new start symbol, and $\mathcal N'$ and $\mathcal P'$ are as follows.
Let $\hat{\mathcal N}$ be the set of symbols obtained from $\mathcal N$ by placing a hat on them. 
Then the disjoint union $\mathcal N' = \mathcal N\cup \hat{\mathcal N}\cup\{S_0, \tilde{\s}\}$ is the new set of nonterminals.

The productions $\mathcal P'$ are  as follows:
\begin{enumerate}
\item  $\s_0 \ra \s$, $\s_0 \ra \tilde{\s}^{\$}$, $\hat{\s}^{\$}\ra \e$ 
\item for each $f\in \mathcal I$, a production $\tilde{\s} \ra \tilde{\s}^f$
\item for each production $\A \rightarrow a$ in $\mathcal P$, a production $\tilde{\s} \rightarrow a\hat{\A} $
\item for each production $A\ra B^f$ in $\mathcal P$, a production
$\hat{B}^f \ra \hat{A}$
\item for each production $\A^f \ra \B$ in $\mathcal P$, a production 
$\hat{\B} \ra \hat{\A}^f$
\item for each production $\A \ra \B\C$ in $\mathcal P$,  productions
$\hat{\B} \ra \C \hat{\A}$ and $\hat{\C} \ra \hat{\A} \B$
\end{enumerate}

Note that the new grammar is no longer in normal form.

Informally, the new grammar operates as follows.
Let $w = w_1 w_2\in L$ and suppose we wish to produce $w_2 w_1$.
If a derivation starts with $\s_0 \ra \s$, then the word produced is some word from $L$.
(This corresponds to the case when one of the $w_i$ is empty.)
Otherwise derivations start with $\s_0 \ra \tilde{\s}^\$$, followed by some sequence of productions 
$\tilde{\s} \ra \tilde{\s}^f$,
building up a flag word on $\tilde{\s}$. This is how we nondeterministically guess the flag label 
$\gamma$ on the second last node  of the path-skeleton.
After this we apply a production $\tilde{\s} \rightarrow a\hat{\A} $,
where $a$ is the first letter of $w_2$ (labelling the end leaf of the path-skeleton) 
and $A$ is the non-terminal labelling the 
second last vertex of the path-skeleton.  
Note that the flag label $\gamma\$$ is transferred to $\hat{\A}$. 
After this point, productions of types 4, 5, and 6 are applied to simulate going in 
reverse along the path-skeleton, at each step producing a sentential form with 
exactly one hatted symbol.
The only way to remove the hat symbol is to apply the production $\hat{\s}^{\$}\ra \e$. 
Observe that all flags on nonterminals in a derivation starting from $\s_0 \ra \tilde{\s}^{\$}$  
are words in $\mathcal I^*\$$, and since $\$$ is always at the right end of a flag it 
does not interfere with any productions from $\mathcal P$, so in particular rules $A\ra a$ 
to the sides of the path-skeleton produce the same strings of terminals as they do in $\Gamma$.

We will show by induction on $n$ that in this new grammar, if $\A, \A_1, \ldots, \A_n\in \mathcal N$ then 
\begin{equation}\A^w \Ra \A^{w_1}_1\ldots  \A^{w_i}_i \ldots \A^{w_n}_n \label{deriv:forward} \end{equation} if and only if 
\begin{equation}\hat{\A}^{w_i}_i \Ra \A^{w_{i+1}}_{i+1} \ldots \A^{w_n}_n \hat{\A}^{w} \A^{w_1}_1 \ldots \A^{w_{i-1}}_{i-1}  
\label{deriv:rotated} \end{equation} for all $1\leq i\leq n$.

To see why this will suffice, suppose first that \[ \s \Ra \A^{w_1}_1\ldots 
\A^{w_{i-1}}_{i-1}  \A^{w_i}_i \A^{w_{i+1}}_{i+1} \ldots \A^{w_n}_n  \ra 
\A^{w_1}_1\ldots \A^{w_{i-1}}_{i-1} a \A^{w_{i+1}}_{i+1} \ldots \A^{w_n}_n \] in the original grammar $\Gamma$. 
So $\A_i\ra a $ is in $\mathcal P$.
Then in the new grammar
\[
 \s_0 \Ra \tilde{\s}^{w_i\$}  \ra a \hat{\A}^{w_i\$}_i 
  \Ra a \A^{w_{i+1}\$}_{i+1} \ldots \A^{w_n\$}_n \hat{\s}^{\$} \A^{w_1\$}_1 \ldots \A^{w_{i-1}\$}_{i-1} \\
  \ra a \A^{w_{i+1}\$}_{i+1} \ldots \A^{w_n\$}_n \A^{w_1\$}_1 \ldots \A^{w_{i-1}\$}_{i-1}.
\]
Each $A^{w_j\$}_j$ produces exactly the same set of words in $\Gamma'$ as $A_j^{w_j}$ produces in $\Gamma$.
Hence every cyclic permutation of a word in $L$ is in the new language.

Conversely, suppose $\s_0 \Ra a \B^{v_1}_1 \ldots \B^{v_n}_n$ and that this subderivation does not start with  $\s_0 \ra \s$.
Then the subderivation begins with $\s_0 \ra \tilde{\s}^{\$} \Ra \tilde{\s}^{u} \ra a \hat{\A}^{u}$ for some $u\in {\mathcal I}^* \$$, $\A\in \mathcal N$.
Once a `hatted' symbol has been introduced, the only way to get rid of the hat is via the production $\hat{\s}^{\$} \ra \e$.
Hence we must have $\hat{\A}^{u} \Ra \B^{v_1}_1 \ldots \B^{v_j}_j \hat{\s}^{\$}  \B^{v_{j+1}}_{j+1} \ldots \B^{v_n}_n$ for some ${0\leq j\leq n}$
(with the factor before or after $\hat{\s}$ being empty if $j=0$ or $j=n$ respectively).

But then 
\[
\s^{\$} 
 \Ra \B^{v_{j+1}}_{j+1} \ldots \B^{v_n}_n \A^u \B^{v_1}_1 \ldots \B^{v_j}_j \\
\ra \B^{v_{j+1}}_{j+1} \ldots \B^{v_n}_n a \B^{v_1}_1 \ldots \B^{v_j}_j 
\]
and so if a word is produced by the new grammar, some cyclic permutation of that word is in $L$.

We finish by giving the inductive proof of the equivalence of \eqref{deriv:forward} and \eqref{deriv:rotated}. For the case $n=1$, 
the productions of type 5 and 6 in the definition of the grammar for $\cyc(L)$ show that $\A^w \Ra \B^u$ if and only if $\hat{\B}^{u} \Ra \hat{\A}^{w}$. 
For the case $n=2$, we have $\A^w \Ra \B^u \C^v$ if and only if at some point in the parse tree, 
we see a subtree labeled
 $\X^t \ra \Y^t \Z^t$, 
with $\A^w \Ra \X^t$, $\Y^t \Ra \B^u$ and $\Z^t \Ra \C^v$.  The productions in these last three subderivations are all of the form
$D \ra {E}^f$ or ${D}^f \ra {E}$, so they are equivalent to 
$\hat{\X}^t \Ra \hat{\A}^w$, $\hat{\B}^u \Ra \hat{\Y}^t$ and $\hat{\C}^v \Ra \hat{\Z}^t$.  
Also $\X \ra \Y \Z$ if and only if $\hat{\Y} \ra \Z \hat{\X}$ and $\hat{\Z} \ra \hat{\X} \Y$.
Putting these together, we have $\A^w \Ra \B^u \C^v$ if and only if
\[ \hat{\B}^u \Ra \hat{\Y}^t \ra \Z^t \hat{\X}^t \Ra \C^v \hat{\A}^w \]
and
\[ \hat{\C}^v \Ra \hat{\Z}^t \ra \hat{\X}^t \Y^t \Ra \hat{\A}^w \B^u, \]
as required.

Now for $n>2$, suppose our statement is true for $k<n$.
Then $\A^w \Ra \A^{w_1}_1 \A^{w_2}_2 \ldots \A^{w_n}_n$ if and only if for each $1\leq i\leq n$
there are $\X_i, \Y_i, \Z_i\in {\mathcal N}$ and $t\in {\mathcal I}^*$ such that $\X_i\ra \Y_i \Z_i$ and 
for some $1\leq j\leq n$ either
\[ \A^w \Ra \A^{w_1}_1 \ldots \A^{w_{i-1}}_{i-1} \X_i^t \A^{w_j}_j \ldots \A^{w_n}_n, \]
with $\Y_i^t \Ra \A_i^{w_i}$ and $\Z_i^t \Ra \A_{i+1}^{w_{i+1}} \ldots \A_{j-1}^{w_{j-1}}$, or
\[ \A^w \Ra \A^{w_1}_1 \ldots \A^{w_j}_j \X_i^t \A^{w_{i+1}}_{i+1} \ldots \A^{w_n}_n, \]
with $\Y_i^t \Ra \A_{j+1}^{w_{j+1}} \ldots \A_{i-1}^{w_{i-1}}$ and $\Z_i^t \Ra \A_i^{w_i}$.

We will consider only the second of these, as it is the slightly more complicated one and the first is 
very similar.  
The right hand side of the displayed subderivation has fewer than $n$ terms, so by our assumption, this
subderivation is valid if and only if 
\[ \hat{\X}_i^t \Ra \A^{w_{i+1}}_{i+1} \ldots \A^{w_n}_n \hat{\A}^w \A^{w_1}_1 \ldots \A^{w_{j}}_{j}. \]
But this, together with $\Y_i^t \Ra \A_{j+1}^{w_{j+1}} \ldots \A_{i-1}^{w_{i-1}}$ and $\Z_i^t \Ra \A_i^{w_i}$,
is equivalent to the existence of a derivation
\[ \hat{\A}^{w_i}_i \Ra \hat{\Z}_i^t \ra \hat{\X}_i^t \Y_i^t \Ra \A^{w_{i+1}}_{i+1} \ldots \A^{w_n}_n \hat{\A}^w \A^{w_1}_1 \ldots \A^{w_{i-1}}_{i-1}\]
such that $\hat{X}_i^t\Ra \A^{w_{i+1}}_{i+1} \ldots \A^{w_n}_n \hat{\A}^w \A^{w_1}_1 \ldots \A^{w_{j}}_{j}$ and $\Y_i^t \Ra \A_{j+1}^{w_{j+1}} \ldots \A_{i-1}^{w_{i-1}}$.
Here, $\hat{\A}^{w_i}_i \Ra \hat{\Z}_i^t$ follows from the equivalence of  \eqref{deriv:forward} and \eqref{deriv:rotated} for $n=1$.
\end{proof}

\section{Concluding remarks}

The results in this paper raise the question whether for an indexed language $L$ the language $C^k(L)$ is indexed as well, or if not, to which class of languages (within context-sensitive) it belongs. 

A consequence of our main result (Theorem~\ref{thm:etolmain}) is that permutations of context-free languages 
are indexed (a different proof of this based on parse trees can be found in \cite{BCE}).
It would be interesting to consider the possible extension of this result to the OI- and IO-hierarchies
(\cite{MR666544}, \cite{MR864744}) of languages built out of automata or grammars that extend 
the pushdown automata and indexed grammars, respectively. They define level-$n$ grammars inductively, 
allowing the flags at level $n$ to carry up to $n$ levels of parameters in the form of flags. 
Thus level-$0$ grammars generate context-free languages, and level-$1$ grammars produce indexed languages.
We conjecture that the class of level-$n$ languages is closed under cyclic closure, and also that 
if $L$ is a level-$n$ language then $C^k(L)$ is a level-$(n+1)$ language.

\nocite{*}
\bibliographystyle{abbrvnat}
\bibliography{refsDMTCS}
\label{sec:biblio}

\end{document}